\documentclass[a4paper,twocolumn,11pt]{quantumarticle}
\pdfoutput=1
\usepackage[utf8]{inputenc}
\usepackage[english]{babel}
\usepackage[T1]{fontenc}
\usepackage{amsmath}
\usepackage{tikz}
\usepackage{graphicx}
\usepackage{amsthm} 
\usepackage{hyperref}

\newtheorem{lemma}{Lemma}

\begin{document}

\title{Dissipative State Engineering of Complex Entanglement with Markovian Dynamics}

\author{Manish Chaudhary}
\affiliation{Institut de Physique Nucléaire, Atomique et de Spectroscopie, CESAM, University of Liège, B-4000 Liège, Belgium}
\orcid{0000-0001-7994-7851}
\email{manish.phys123@gmail.com}
\maketitle

\begin{abstract}
  Highly multipartite entangled states play an important role in various quantum computing tasks.
  We investigate the dissipative generation of a complex entanglement structure as in a cluster state through engineered Markovian dynamics in the spin systems coupled via Ising interactions. Using the Lindblad master equation, we design a projection based dissipative channel that drives the system toward a unique pure steady state corresponding to the desired cluster state. This is done by removing the contribution of the orthogonal states.  
  By explicitly constructing the Liouvillian superoperator in the full $2^N$-dimensional Hilbert space, we compute the steady-state density matrix, the Liouvillian spectral gap, entanglement witness and the fidelity with respect to the ideal cluster state. The results demonstrate that the cluster state emerges as the steady state when the engineered Liouvillian dissipation dominates over the local Ising interaction between spins. Moreover, we find that the fidelity and Liouvillian spectral gap is relatively insensitive to the system size once the saturation dissipation has been achieved that scales linearly with the qubit number. This analysis illustrates a physically realizable path towards steady-state entanglement generation in the spin systems using engineered dissipation. 
\end{abstract}

\section{Introduction}

Entanglement generation and stabilization~\cite{RevModPhys.81.865,RevModPhys.80.517} constitute central challenges in quantum information processing~\cite{PhysRevLett.70.1895,PhysRevLett.68.3121,ladd2010quantum}. Preparation of entangled pure quantum states is of interest in the context of both condensed matter physics~\cite{PhysRevB.95.094302,PhysRevLett.126.040602,PhysRevB.83.094431} and quantum information~\cite{RevModPhys.86.153,RevModPhys.90.035005}. In condensed matter physics, entangled states represent ground states of the physical Hamiltonian~\cite{laflorencie2016quantum}, while in quantum information, qubit entangled states can act as a resource for quantum computing~\cite{Nielsen:2010,ladd2010quantum}. For instance, cluster states represent a special class of highly entangled multipartite graph state~\cite{PhysRevLett.86.910} which form the resource backbone for measurement-based quantum  computation~\cite{walther2005experimental,PhysRevA.68.022312,PhysRevLett.86.5188,PhysRevLett.112.120504}. It provides a novel platform for studying computational model using single qubit measurements~\cite{PhysRevLett.86.5188}, non locality~\cite{PhysRevLett.95.120405,santos2023scalable} and quantum error correction~\cite{bell2014experimental,PhysRevA.65.012308}.  The generation of cluster states has been demonstrated in various experiments using photons~\cite{yokoyama2013ultra,schwartz2016deterministic,lu2007experimental}, neutral atoms~\cite{mandel2003controlled,cooper2024graph}, superconducting qubits~\cite{PhysRevLett.97.230501,petrovnin2023generation} and trapped ions~\cite{PhysRevLett.111.210501}.  Cluster states in higher dimensions for qudits has also been proposed~\cite{reimer2019high,roh2025generation}.

Generally, entangled states are prepared through coherent unitary dynamics~\cite{RevModPhys.81.865,blatt2008entangled,PRXQuantum.5.030344} followed by projective measurements. 
However, recent methods based on dissipative engineering provide an alternative route~\cite{verstraete2009quantum,lin2025dissipative}, in which, for a given many-body Hamiltonian, one can design artificial reservoirs to steer the system to a desired pure entangled steady state~\cite{PhysRevA.78.042307,verstraete2009quantum,PhysRevA.65.010101}. The interaction of a quantum system with its surrounding environment in a controlled manner is not necessarily a disadvantage. Recently it is increasingly useful in various applications in quantum information, such as preparation of ground states of stabilizer codes and long range entanglement~\cite{PhysRevResearch.2.033347,PRXQuantum.3.040337,PhysRevResearch.6.033147}, quantum error correction~\cite{PhysRevLett.133.030601} and quantum computation~\cite{verstraete2009quantum,sannia2024dissipation}.  

In this paper, we have proposed dissipative state preparation of a one-dimensional cluster state of spin qubits arranged in a linear lattice connecting with nearest neighbor interactions~\cite{PhysRevA.77.062330,PhysRevA.77.012321} in an external transverse field. We consider a model in which each qubit is weakly coupled to a memoryless reservoir that defines a Markovian dissipative regime. Various works based on Markovian dynamics~\cite{PhysRevA.78.042307} have recently been performed to generate entangled states. This include atom-cavity entanglement~\cite{PhysRevA.96.052311,greve2022entanglement} and entanglement through dissipative phase transitions~\cite{luo2017deterministic,cho2017quantum}. But none of the works have demonstrated the preparation of the special class of the cluster state yet using the dissipative engineering.

In our work, we have used the interplay between Ising interactions and dissipation to generate a steady state that resembles the desired state. The Ising interaction creates entanglement in the spin chain as studied in~\cite{PhysRevLett.86.910,PhysRevA.77.062330,PhysRevA.77.012321} while dissipation directs this vital entanglement towards the cluster state correlations in a deterministic sense in the large dissipation limit. To achieve this, we have employed state-projection based construction to Lindblad operator that can be viewed as equally as picking up the right energy state corresponding to the cluster state. This approach differs from the earlier works in~\cite{PhysRevE.73.016139,gonzalez2024tutorial} where the projection operators are used to map onto correlated state of the system-bath to derive Master equation, in our work we have addressed this construction at the level of the Lindlad jump operators itself. 
Our scheme has a potential to implement it on an experimental platform, moreover it can be scaled to generate large size cluster states. The scheme can be extended to prepare cluster state in higher dimensions and we present a prototype for square lattice cluster states as an example. 

The paper is structured as follows: in Sec. \ref{sec2} we review theoretical framework for cluster state entanglement and Markovian dissipation dynamics and introduce the model Hamiltonian. In Sec. \ref{sec3} we introduce the dissipative protocol and discuss the construction of Liouvillian such that the cluster state is the steady state for the model. In Sec. \ref{sec4} we evaluate the performance of our protocol numerically and calculate the steady states and its response in the small and large dissipation limit. Moreover we check the fidelity of the steady state with the cluster state and find its scaling with the finite system size. We also calculate Liouvillian gap and spectrum to check the emergence of the steady state. We characterize the multipartite entanglement using witness operator. In Sec. \ref{sec5} we extend the protocol to demonstrate the robust generation of square cluster state in high dimension. Then in Sec. \ref{sec6} we discuss the experimental implementation of our model and discuss its viability. Finally in Sec. \ref{sec7} we summarize our findings and provide an outlook to the proposed protocol.


\section{Theoretical Framework}
\label{sec2}
In this section we review the notion of complex entanglement structure for the cluster states that has profound applications in measurement based quantum computing. We also introduce the dissipation process defined by Markovian dynamics.

\subsection{Cluster-State entanglement}
Cluster state is a special class of multipartite graph state~\cite{PhysRevA.69.062311,PhysRevLett.86.910} where the vertices correspond to the spin qubits and the edges represent Ising interactions between qubits as shown schematically in Fig. \ref{fig1}. The interaction can be implemented using controlled phase gate between the two connected qubits.

Mathematically, a one-dimensional $N$-qubit cluster state $|C_N\rangle$ is defined as 
\begin{align}
    |C_N\rangle = \left (\prod_{j,k} CZ_{jk} \right) |+\rangle ^N
    \label{eq:N-Cluster state}
\end{align}
where $j,k$ represents neighboring qubits. 
$CZ$ is the controlled-$Z$ gate such that the state of the $k^{\text{th}}$ target qubit gets a phase of $-1$ when the $j^{\text{th}}$ control qubit is set to $|1\rangle$. 
All qubits are initialized in the superposition state $|+\rangle = \frac{|0\rangle + |1\rangle}{\sqrt{2}}$ such that
\begin{align}
    |+\rangle ^N = \frac{1}{2^{N/2}} \sum_{x=0}^{2^N-1} |x\rangle
\end{align}
where $|x\rangle$ represents the product state of $N$-qubits in the computational basis.

A cluster state is a unique stabilizer state that is not changed under a set of commuting observables~\cite{PhysRevA.69.062311}, i.e., 
\begin{align}
    \mathcal{S}_j |C_N\rangle = |C_N\rangle
    \label{eq:eigenstate}
\end{align}
where the observables $\mathcal{S}_j$ are defined in terms of Pauli matrices \{$X,Y,Z$\} as,
\begin{align}
    \mathcal{S}_j = X_j \prod_{k\in N(j)} Z_k
    \label{eq:stabilizer}
\end{align}
where $N(j)$ represents number of neighboring qubits around $j$.

For instance, a one-dimensional four-qubit cluster state $|C_4\rangle$ is structured as,
\begin{align}
|C_4\rangle = \frac{1}{2}\left( |0000\rangle + |0011\rangle + |1100\rangle - |1111\rangle \right).
\label{eq:4cluster state}
\end{align}
with the stabilizers as
\begin{align}
\mathcal{S}_1 &= X_1 Z_2, \nonumber\\
\mathcal{S}_2 &= Z_1 X_2 Z_3, \nonumber\\
\mathcal{S}_3 &= Z_2 X_3 Z_4, \nonumber\\
\mathcal{S}_4 &= Z_3 X_4.
\end{align}
For a cluster state, the expectation value of each Pauli matrix at each site $k$ follows,
\begin{align}
    \langle C_N| X_k|C_N \rangle & = 0 \nonumber \\
    \langle C_N| Y_k|C_N\rangle & = 0 \nonumber \\
    \langle C_N| Z_k|C_N \rangle & = 0
\end{align}
We discuss the generation of $N$-qubit linear cluster state in Sec. \ref{sec3} using dissipative process and generalize it to the square lattice cluster state in Sec. \ref{sec5}.
\begin{figure}[t]
  \centering
  \includegraphics[width=\linewidth]{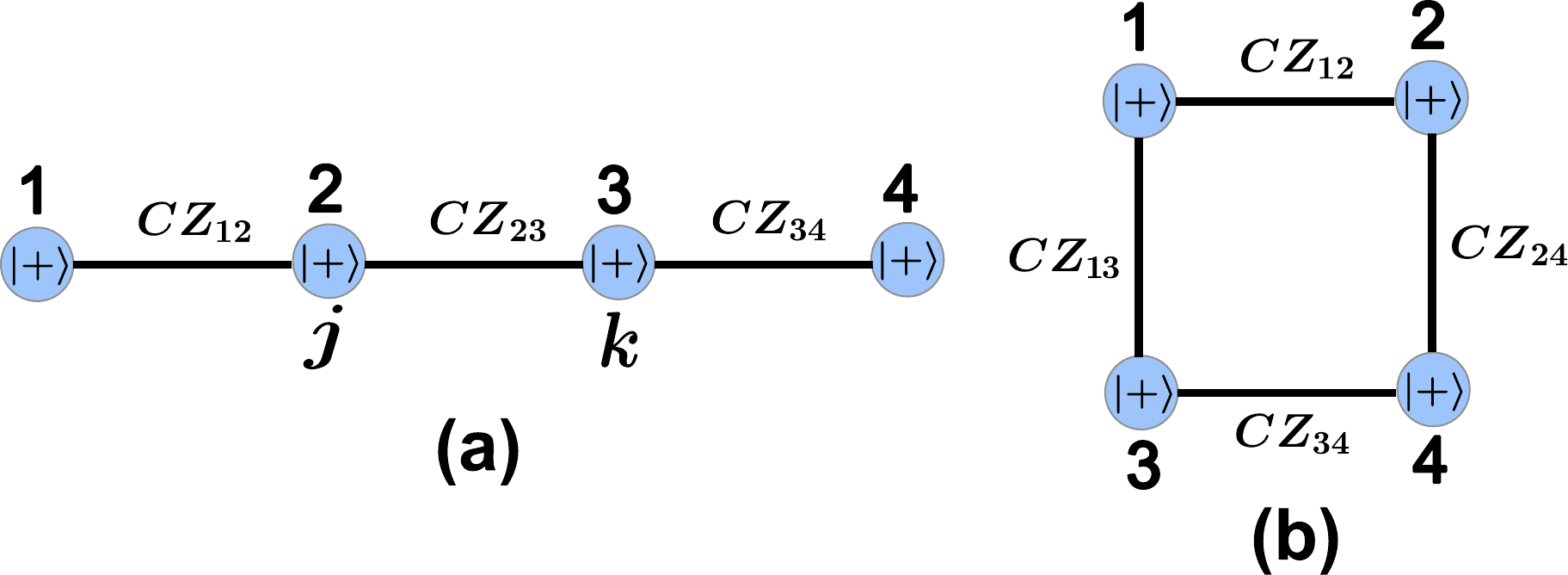}
  \caption{Schematic representation of (a) the linear cluster state in one dimension~\eqref{eq:N-Cluster state}; (b) the square lattice cluster state in two dimension~\eqref{eq:sqaure_cluster state}. Each qubit (vertex) is initialized in the superposition state $|+\rangle$ with the controlled phase CZ gate interaction (edges) between the connecting qubits j,k.}
  \label{fig1}
\end{figure}

\subsection{Physical model and Markovian dynamics}

When the system is coupled to a global memoryless reservoir, time dynamics of the system is usually governed by quantum Markov process obeying Master equation~\cite{schlosshauer2007decoherence} as
\begin{align}
\dot{\rho} = \mathcal{L}[\rho] = -i[H,\rho]+ \gamma \left( L \rho L^\dagger - \frac{1}{2}\{L^\dagger L,\rho\} \right) 
\label{eq:masterequation}
\end{align}
where $\rho$ is the system density matrix, $H$ represents the system Hamiltonian and $L$ is Lindblad jump operator with the corresponding decay rates given by $\gamma \geq 0$. We consider the dissipation rate to be identical for all qubits.

\begin{figure}[t]
  \centering
  \includegraphics[width=\linewidth]{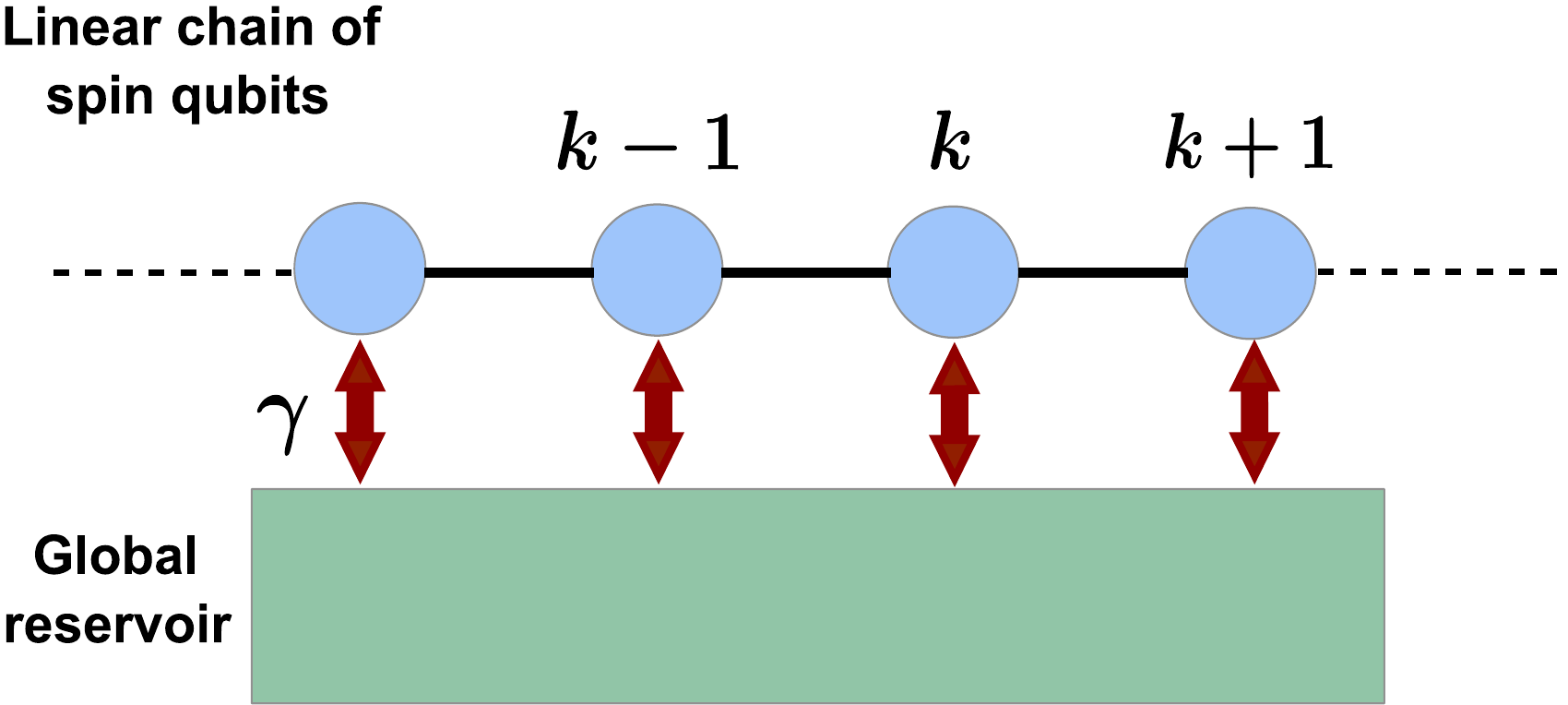}
  \caption{Dissipative model \eqref{eq:masterequation} for a linear chain of $N$ spin qubits: each k$^\text{th}$ qubit interacts with its nearest neighbors (k-1$^\text{th}$ and k+1$^\text{th}$) through Ising coupling \eqref{eq:ham} (Coherent dynamics). Each qubit is coupled to a single memoryless reservoir (dissipation dynamics) in an identical way such that the desired cluster state is the steady state of the Liouvillian \eqref{eq:lemma1}.  }
  \label{fig2}
\end{figure}
We consider a physical system of a linear chain of locally addressable $N$ spin qubits with nearest neighboring Ising interactions~\cite{PhysRevA.77.062330,PhysRevA.77.012321} in the presence of a transverse magnetic field as shown in Fig. \ref{fig2}. Each qubit is locally coupled to the dissipative environment in an identical way. The Hamiltonian is defined as,
\begin{align}
    H = g \sum_k Z_k Z_{k+1} + h \sum_k X_k
    \label{eq:ham}
\end{align}
where $g$ represents the coupling between nearest neighboring qubits and $h$ represents the magnitude for the driving pulse or transverse magnetic field. With $\gamma_k=0$, the ground state interpolates between the ordered Ising phase and the spin polarized phase by tuning the ratio $g/h$. The dynamics of this model is described in details in Ref.~\cite{PhysRevA.77.062330,PhysRevA.77.012321,PhysRevLett.99.177210} without dissipation.

In order to prepare cluster entangled state through a dissipative process, the basic requirement is to identify quantum reservoirs such that one can design system-reservoir coupling for the master equation \eqref{eq:masterequation}. Then in this case, the desired pure state of a many-body system is obtained as the unique steady state $\rho_{s}$ for the Liouvillian $\mathcal{L}$. It implies that any initial system density matrix evolves to the unique steady state density matrix in the long times,
\begin{align}
    \rho_{\text{initial}} \xrightarrow{t \to \infty} \rho_s
\end{align}

Liouvillian properties are crucial to understand the dynamics of the physical system~\cite{PhysRevA.89.022118}.
By computing Liouvillian spectra~\cite{minganti2018spectral}, one can find steady state density matrix $\rho_{s}$. Diagonalization of $\mathcal{L}$ yields its eigenvalues and eigenvectors.
It is calculated by solving the eigenvalue equation
\begin{align}
    \mathcal{L}(\rho) = \lambda \rho
    \label{eq:eigenval}
\end{align}
One finds the steady state density matrix $\rho_s$ corresponding to the eigenvalue $\lambda_0=0$ as described theoretically in~\cite{PhysRevA.89.022118}.
With suitable choice for Lindblad operators, one can design the steady state of Liouvillian to be target state. 

Moreover we study the Liouvillian gap in the dissipative model \eqref{eq:masterequation} to study the relaxation of the dissipative system towards the steady state for a finite-size system. This is calculated from the spectral properties of Liouvillian \cite{minganti2018spectral}.  
We define the Liouvillian gap $\Delta$, as the difference between two absolute real eigenvalues of the Liouvillian matrix as,
\begin{align}
   \Delta= |\text{Re}[\lambda_1]| - |\text{Re}[\lambda_0]|
    \label{eq:liou_gap}
\end{align}
where $\lambda_1 \neq 0$ is the next lowest eigenvalue of the Liouvillian matrix \eqref{eq:eigenval}.

In the next section, we discuss the full construction of Liouvillian such that  cluster state is the unique steady state with the eigenvalue zero.


%

\section{Dissipative Protocol for Entanglement generation}
\label{sec3}
In this section, we discuss the construction for Lindblad operators that can be used to generate cluster states of any dimension as a steady state.
The main idea behind this construction is to use the projection mechanism that maps any
orthogonal state onto a target state that contains cluster state like correlations between qubits originating from interactions~\cite{PhysRevLett.86.910,PhysRevA.77.062330}. Our objective is to employ Liouvillian superoperator that relaxes the system towards the cluster state only. It is motivated by splitting the Hilbert space of density matrix into two parts~\cite{gonzalez2024tutorial} so that dissipative dynamics picks up the target state space.

\begin{figure}[t]
  \centering
  \includegraphics[width=\linewidth]{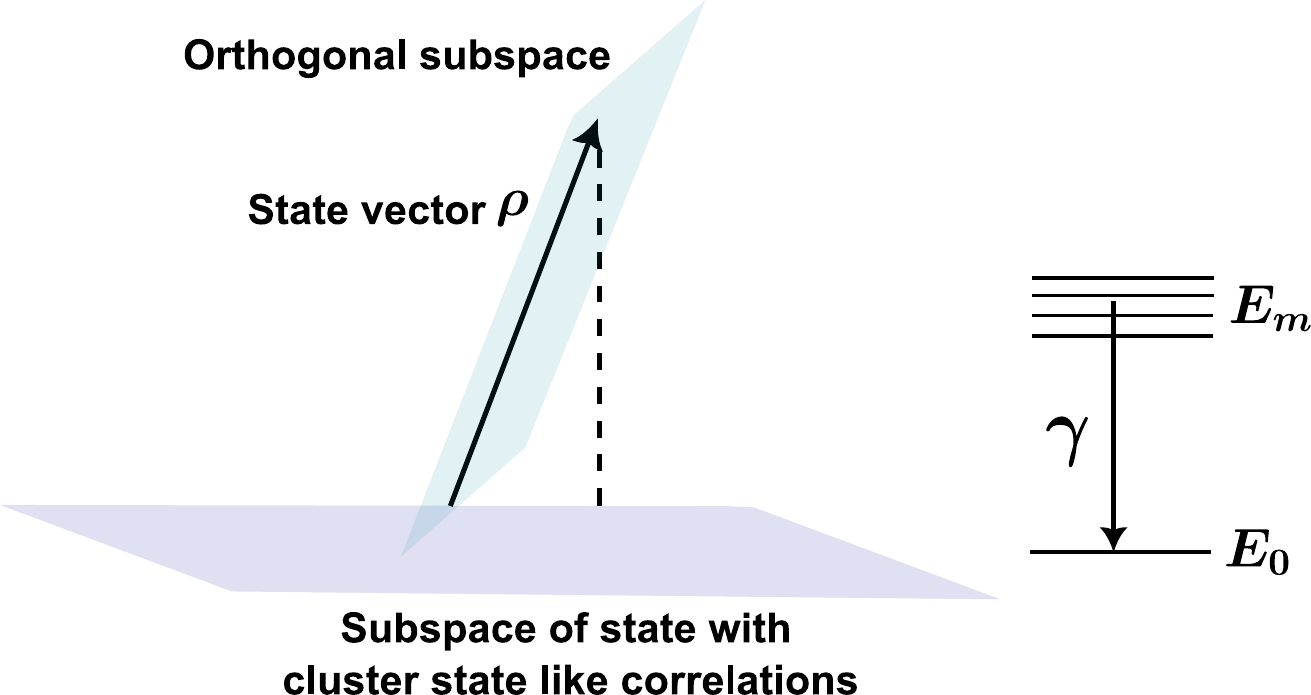}
  \caption{Schematic illustration of the construction of the Lindblad jump operator in Eq.~\eqref{eq:Lindbladope}. The dissipative dynamics projects an arbitrary state $\rho$ from the full Hilbert space onto the subspace spanned by the target cluster state. Equivalently, the action of the jump operator can be interpreted as driving the system toward a specific energy eigenstate $E_0$, as depicted in the adjacent panel, thereby stabilizing the cluster state as the unique steady state of the evolution.}
  \label{fig3}
\end{figure}
\subsection{Liouvillian construction}
If $\mathcal{H}$ is the Hilbert space of the system such that its dimension is $d=\text{dim}(\mathcal{H})$, an orthogonal complement of span$\{|C_N \rangle\}$ is a subspace of dimension $d-1$.
Let the set $\{|\phi_m\rangle\}$ represents the basis of the subspace orthogonal to the cluster state $|C_N\rangle$ such that 
\begin{align}
    \langle \phi_m |C_N\rangle =0 
    \label{eq:orthogonality}
\end{align}
and $\langle \phi_{m'}|\phi_m\rangle = 0$. There is no unique choice for $|\phi_m\rangle$, any orthonormal basis will be a sufficient choice which can be obtained using Gram–Schmidt process from the computational basis. An example for this construction in the case of two-qubit linear cluster state $|C_2\rangle$ is given in Appendix A.

In terms of the projection operators, we can define two subspaces
\begin{align}
    \hat{P}_1  & = |C_N\rangle \langle C_N | \nonumber \\
     \hat{P}_2  & = \sum_m |\phi_m\rangle \langle \phi_m |
\end{align}
so that it projects the density matrix $\rho$ onto the corresponding subspace  $\hat{P}_i \rho$.

With this, we define non-Hermitian Lindbald operators as,
\begin{align}
    L_m = |C_N\rangle \langle \phi_m|
    \label{eq:Lindbladope}
\end{align}
where the index $m$ represents the number of the basis states in $(2^{N}-1)$-dimensional orthogonal subspace.
The action of superoperator $L_m \rho$ is to pump any orthogonal component within $\rho$ into the desired state $|C_N\rangle$ as this is the unique steady state for a given Liouvillian $\mathcal{L}$ construction. This definition can also be viewed as the \textit{projection} based Lindblad construction. In terms of physical implementation, one can view this as projecting onto the energy state $E_0$ corresponding to the spin cluster state 
\begin{align}
    |E_0 \rangle \langle E_m|
\end{align}
while all the orthogonal states with energies $E_m$ are smeared out. The construction is shown schematically in Fig. \ref{fig3}.
We now list important properties of the Liouvillian. 

\begin{lemma}
    A cluster state $|C_N\rangle\langle C_N|$ is a unique steady state for the Liouvillian $\mathcal{L}$ in the strong dissipation limit.
    \begin{align}
        \mathcal{L}(|C_N\rangle \langle C_N |) = 0
        \label{eq:lemma1}
    \end{align}
\end{lemma}

\begin{proof}
    We start with the Eq. \eqref{eq:masterequation} and evaluate $\mathcal{L}$ on $\rho=|C_N\rangle \langle C_N|$ term-by-term, 

    For every $m$,
    \begin{align}
        L_m \rho L_m^\dagger & = L_m |C_N\rangle \langle C_N| L_m^\dagger \nonumber \\
        & = (|C_N\rangle \langle \phi_m| C_N\rangle) (\langle C_N| \phi_m\rangle |C_N\rangle) 
    \end{align}
    Using the properties for basis vectors \eqref{eq:orthogonality}, it evaluates to
     \begin{align}
        L_m \rho L_m^\dagger = 0
      \end{align}
      Also,
      \begin{align}
         L_m^\dagger L_m & = (|C_N\rangle \langle \phi_m|)^\dagger (|C_N\rangle \langle \phi_m|) \nonumber \\
         & = |\phi_m\rangle \langle \phi_m|
      \end{align}
       Therefore, the anticommutator sum is
       \begin{align}
          \sum_m \frac{1}{2}\{L_m^\dagger L_m,\rho\} & = \frac{1}{2}  \left\{\sum_m|\phi_m\rangle \langle \phi_m|,|C_N\rangle \langle C_N|\right\}
       \end{align}
Using the completeness relation,
       \begin{align}
         \sum_m |\phi_m\rangle \langle \phi_m| + |C_N\rangle  \langle C_N| = I
       \end{align}
       it follows that 
       \begin{align}
         \left\{I-|C_N\rangle \langle C_N|,|C_N\rangle \langle C_N|\right\} =0
       \end{align}
       
       Hence, the anti commutator term also vanishes as
       \begin{align}
          \sum_m \frac{1}{2}\{L_m^\dagger L_m,\rho\} =0
       \end{align}
       For the Hamiltonian \eqref{eq:ham}, as the operators $Z_kZ_{k+1}$ and $X_k$ anti commute with the stabilizers \eqref{eq:stabilizer} and hence this generates a state orthogonal to the cluster state $|C_N\rangle$.
Therefore one finds that in the strong dissipation limit $|\frac{g}{h}|<<\gamma$,
\begin{align}
    [H,\rho]  \approx 0
\end{align}
hence this completes the proof \eqref{eq:lemma1}.

To provide uniqueness of the steady state, we characterize the Kernel of the Liouvillian. 

Let $\rho_s$ be any steady state density matrix in the strong dissipation limit,
\begin{align}
    \mathcal{L}(\rho_s) = 0
\end{align}
Since each term in the Lindbladian is positive semi definite that requires
\begin{align}
    L_m \rho_s L_m^\dagger = 0 
\end{align}
for each $m$. Using \eqref{eq:Lindbladope} it is found that
\begin{align}
    \langle \phi_m | \rho_s |\phi_m\rangle = 0
\end{align}
which implies that $\rho_s$ has zero population in every state $|\phi_m\rangle$. Since the set $\{|\phi_m \rangle\}$ spans the entire subspace orthogonal to $|C_N\rangle$, so the only possibility is the occupancy in the cluster state only. Hence 
\begin{align}
    \text{dim}[\text{ker}  (\mathcal{L})] = 1
\end{align}
which is true for $|\frac{g}{h}|<<\gamma$.
However for weak dissipation, the dynamics are dominated by nearest neighbor interactions that leads to the existence of other steady states than the cluster state. 
\end{proof}

\section{Numerical simulation and performance}
\label{sec4}
In this section we analyze the performance of the proposed protocol to generate cluster state with dissipative engineering. At first we calculate mean field dynamics in the large $N$ limit and obtain stable steady states. We also calculate the dynamics of the exact time evolved state and study spin expectation values, Liouvillian gap and fidelity of the final state to check the convergence of Markovian dynamics towards cluster state preparation. We study the emergence of the critical phenomenon with the system size $N$ by performing exact time evolution of the system's density matrix (\ref{eq:masterequation}) using the developed software based on the Python package~\cite{PhysRevA.98.063815}.

\subsection{Mean field dynamics and stable steady states}
We obtain the steady state properties of this model \eqref{eq:masterequation} in the mean field (MF) approximation where the total density matrix is approximated as
\begin{align}
    \rho \approx \underset{k}{\otimes} \rho_k
    \label{eq:MF_density matrix}
\end{align}
and the product of Pauli matrix is factorized as,
\begin{align}
    \langle O_j O_k \rangle \approx \langle O_j \rangle  \langle O_k \rangle
    \label{eq:MF_observable}
\end{align}

Mean field equations for the averaged spin observables are obtained as,
\begin{align}
    \dot{J}^x &=  -4g J^y J^z  -\gamma J^x \nonumber \\
    \dot{J}^y & =  4g J^x J^z - 2h J^x -\gamma J^y  \nonumber \\
  \dot{J}^z & =  -2h J^y -\gamma J^z
\end{align}
where  $J^\alpha = \frac{1}{N} \underset{k}{\sum} \langle \alpha_k \rangle $ for $\alpha\in \{X,Y,Z\}$ is the averaged spin value. We choose dimensionless parameters as $(h_{g} = \frac{h}{g},\gamma_{g}= \frac{\gamma}{g})$ to study the dynamics of the system.

Solution of the mean field equations can be broadly categorized into two regimes: in the absence of dissipation $\gamma = 0$, we obtain several steady states. In the limit $|h_g|>>1$ value, i.e. when the transverse field is dominant over nearest neighbor coupling, one obtains $X$-polarized phase such that 
\begin{align}
  (J^x, J^y, J^z)_{s_1} = (\pm 1, 0, 0) .
  \label{eq:stable1}
\end{align}
However when the coupling between nearest neighbor exceeds the transverse magnetic field with $|h_g|<<1$, ordered spin phase is developed with
\begin{align}
    (J^x, J^y, J^z)_{s_2} = \left(\frac{h_g}{2},0,\pm\sqrt{1- \left(\frac{h_g}{2}\right)^2}\right).
     \label{eq:stable2}
\end{align}
Depending on the sign of the coupling $g$, the steady state solutions are characterized by doubly degenerate ground states for $g<0$ as
\begin{align}
    |00\dots 0\rangle, |11\dots 1\rangle
\end{align}
and a finite $h$ lifts the degeneracy leading to a $Z_2$ broken symmetry phase in the thermodynamic limit with $ J^z \neq 0$ and $ J^x \neq 0$. With $g>0$, the degenerate ground states are characterized by 
\begin{align}
    |0101\dots \rangle, |1010\dots \rangle 
\end{align}
with $J^z=0$.

\begin{figure}[t]
  \centering
  \includegraphics[width=\linewidth]{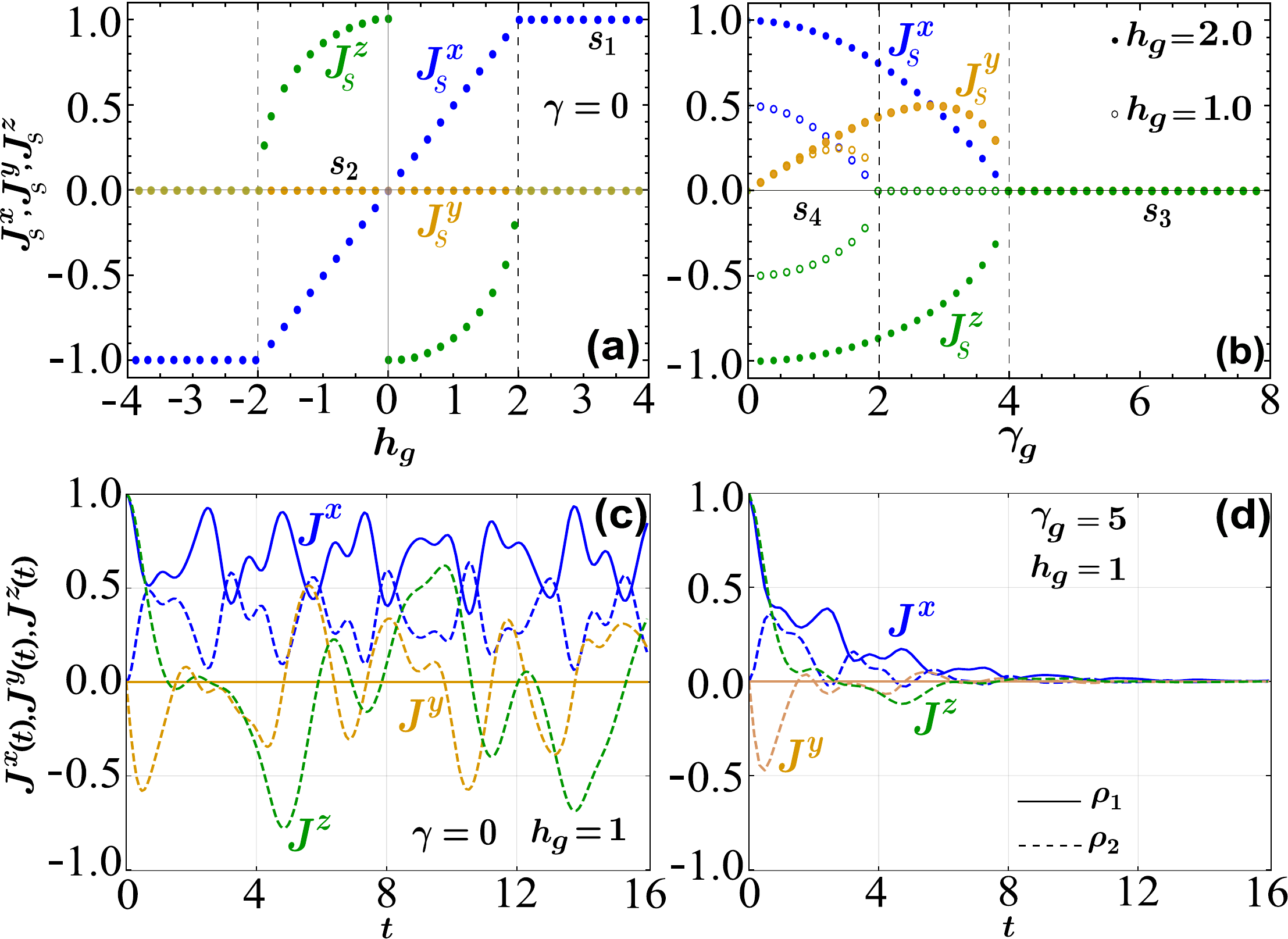}
  \caption{Dynamics of dissipative model \eqref{eq:masterequation} with Markovian environment: Variation of the steady state spin expectation values with the controlling parameters (a) no dissipation ($\gamma=0$); (b) with dissipation at a fixed value of $h_g=2.0$ (closed circle) and $h_g=1.0$ (open circle). Distinct phases are marked as $s_j$. Exact time evolution of the density matrix and spin expectation values \eqref{spin_exp_finite} as a function of time $t$ (c) no dissipation ($\gamma=0$); (d) with dissipation ($\gamma_g=5$) for system size with $N=4$ and two different initial states, $\rho_1(0) = |++++\rangle \langle ++++|$ (solid line) and $\rho_2(0) = |0000\rangle \langle 0000|$ (dashed line).  }
  \label{fig4}
\end{figure}

For a non-zero dissipation $\gamma>0$, the only stable steady state is found to be cluster state in the strong dissipation regime $\gamma_g >2h_g$ with
\begin{align}
    (J^x, J^y, J^z)_{s_3} = (0,0,0)
     \label{eq:stable3}
\end{align}
We conclude that it is not a mixed state by plotting fidelity with the exact cluster state in the next section. 

While in the region $\gamma_g <2h_g$, it follows that
\begin{align}
    (J^x, J^y, J^z)_{s_4} = \frac{\sqrt{4h_g^2-\gamma_g^2}}{8h_g}\left(\sqrt{4h_g^2-\gamma_g^2},\pm \gamma_g,\mp 2h_g\right)
     \label{eq:stable4}
\end{align}

We have plotted mean field steady state spin expectation values in Fig. \ref{fig4}(a),(b) as a function of the controlling parameters.  Fig.  \ref{fig4}(a) shows the existence of possible steady states with no dissipation dynamics. We can see that the spin expectation values $J^x,J^y,J^z$ changes continuously between fixed points \eqref{eq:stable1}-\eqref{eq:stable2}  with a non analytical behavior at $|h_g| =2$. It is noted that a spin polarized phase is generated \eqref{eq:stable1} at larger $h_g$ while an ordered spin phase appears at smaller $h_g$ that is characteristic of the stronger nearest neighbor coupling. On the other hand, dissipation dynamics follows different stable points. One can observe a transition to the cluster state properties near the transition point $\gamma_g = 2h_g$. This transition remains the same irrespective of the sign of $g$.

Fig. \ref{fig4}(c),(d) depicts the time variation of the spin expectation value which is calculated using,
\begin{align}
    J_\alpha (t) = \text{Tr}[\rho(t) J_\alpha ]
    \label{spin_exp_finite}
\end{align}
for a finite $N$ with different initial states $\rho_i(0)$. For zero dissipation case, we observe a finite oscillations in the spin averaged values because of the coherent dynamics. For longer evolution times, it resembles mean field results. In the presence of dissipation, we can see that the system dynamics follow steady state properties for longer time evolution. For larger dissipation and critical coupling, we find that the exact time evolution leads to the cluster state properties satisfying $J^x=J^y=J^z=0$ irrespective of the initial state as seen in Fig. \ref{fig4}(d).

\subsection{Fidelity}
We define the fidelity as the overlap of the steady state density matrix with the cluster state originating from the dissipative process \eqref{eq:masterequation}  
\begin{align}
    F = \frac{\langle C_N|\rho_s|C_N \rangle }{\text{Tr}(\rho_s)}
    \label{eq:fidelity}
\end{align}
Fidelity is $1$ when the steady state is the same as the cluster state.

Fig. \ref{fig5}(a) illustrates the fidelity of the steady state with respect to the target cluster state $|C_N\rangle$ as a function of the dissipation parameter for different system sizes $N$. It is found that for relatively weak dissipation $\gamma_g<<1$, the fidelity remains low as the system dynamics is dominated by the nearest neighbor interactions \eqref{eq:ham} driving it to a different steady state that differs from the cluster state. As $\gamma_g$ is increased, the fidelity rises gradually, reflecting the influence of the engineered dissipation. Beyond a characteristic threshold, referred to as the saturated dissipation strength $\gamma_{\text{sat}}$, the fidelity of the state is maximum $F_\text{sat}$ and exhibits no further increase upon increasing $\gamma_g$. This is remarked by dissipation dominated regime in which the state preparation takes place towards the unique cluster state only. 
\begin{figure}[t]
  \centering
  \includegraphics[width=\linewidth]{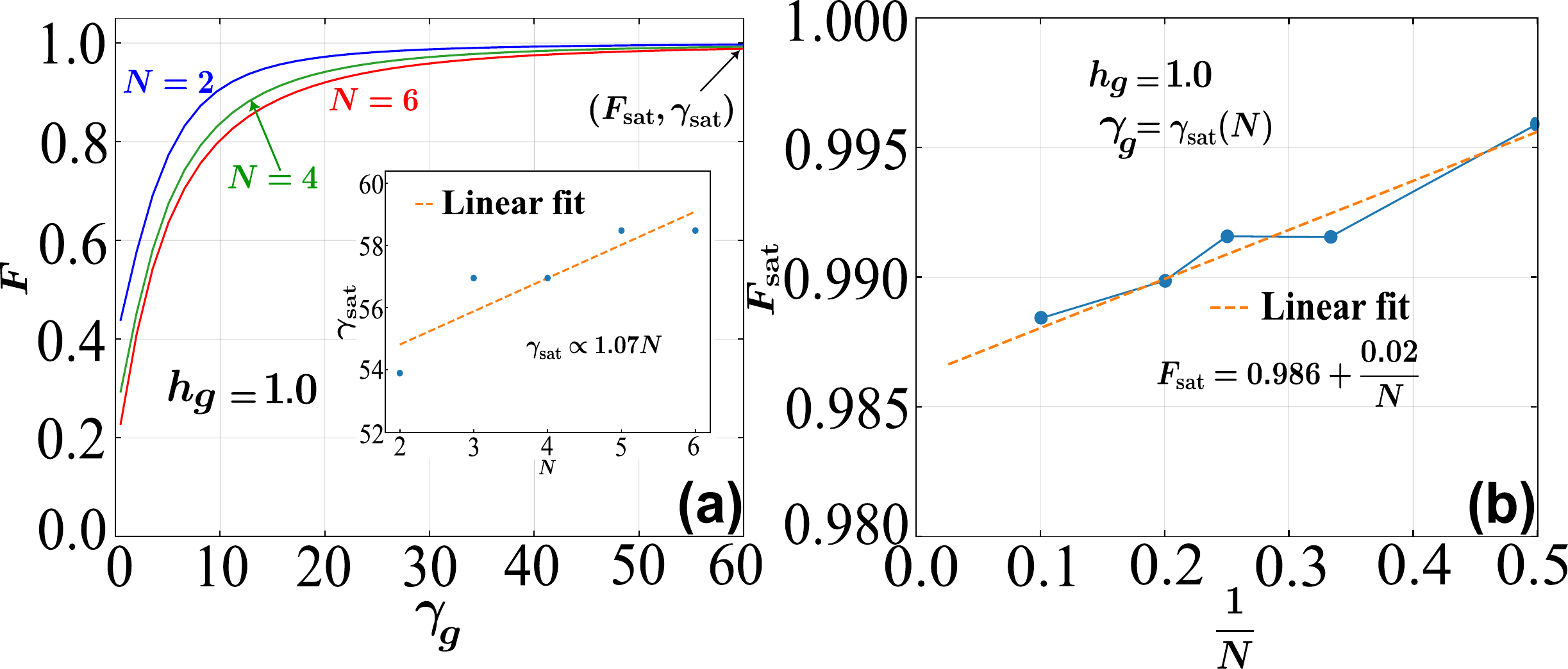}
  \caption{Fidelity \eqref{eq:fidelity} and its scaling with the system size $N$ for the dissipative model \eqref{eq:masterequation} under Markovian dynamics: (a) variation with the dissipative parameter $\gamma_g$ with varying $N$. A zoomed-in plot is shown for the variation of the saturated dissipation $\gamma_{\text{sat}}$ with varying $N$ in the inset capturing a linear scaling; (b) the maximum of the fidelity $F_{\text{sat}}$ is plotted as a function of the system size $N$ showing a power law scaling ($\propto N^{\beta}$). The system parameters are chosen as specified in each panel. For all cases, we choose $h_g=1$. }
  \label{fig5}
\end{figure}

It is observed that with increasing $N$ requires stronger dissipation to stabilize the state such that 
\begin{align}
    \gamma_{\text{sat}}(N+1)>\gamma_{\text{sat}}(N)
\end{align}
To gain maximum fidelity with the cluster state one needs comparatively larger value of $\gamma_{\text{sat}}$ with larger $N$. This happens because of the larger Hilbert space associated with the increasing $N$. We study dependence of the parameter $\gamma_{\text{sat}}$ with the system size $N$ in the inset plot in Fig. \ref{fig5}(a). Using power law we fit the straight line such that
\begin{align}
\gamma_{\text{sat}} & \propto N \nonumber \\
    \gamma_{\text{sat}} & = \gamma_0 + 1.07 N
\end{align}
where $\gamma_0 = 52.68$ is a constant value.

It is technically more logical to deal with the maximum fidelity. In Fig. \ref{fig5}(b) we have plotted the maximum of the fidelity against the system size $N$. Fidelity is larger for small system size at a given $\gamma_{\text{sat}}$. We provide a scaling with $1/N$ by
showing a straight line fit using a power law
\begin{align}
     F_\text{sat} = F_0 + \frac{0.020}{N}
\end{align}
where $F_0 = 0.986$ is an asymptotic value corresponding to the limit $N\to \infty$. We conclude that the maximum fidelity is relatively insensitive to the system size once the limit of the saturation dissipation has been achieved. This shows the potential of our dissipative protocol in the context of the scalability with the number of qubits.

\subsection{Entanglement witness}
To detect multipartite entanglement in the steady state density matrix originating from the dissipative process \eqref{eq:masterequation}, we calculate entanglement witness measure as,
\begin{align}
    W = \eta I_N -|C_N\rangle \langle C_N|
    \label{eq:witness}
\end{align}
where $I_N$ is the identity matrix of dimension $N$ and $\eta=\underset{\rho_{\text{sep}}}{\text{max}}\langle C_N|\rho_{\text{sep}}|C_N\rangle$ is the maximum overlap of the cluster state with any separable state as introduced in Ref.~\cite{PhysRevLett.92.087902}. For a linear cluster state, $\eta=\frac{1}{2}$.

The witness expectation value is defined as,
\begin{align}
    \langle W \rangle = \text{Tr}(W\rho_s)
     \label{eq:witness_expectationvalue}
\end{align}
We choose the above definition in the context of the cluster state as used in Ref.~\cite{lu2007experimental} such that $\langle W \rangle\geq0$ implies separability and $\langle W \rangle<0$ indicates the presence of the multipartite entanglement. Maximum negative expectation value $\langle W \rangle = -\frac{1}{2}$ is attained for the ideal cluster state $|C_N\rangle$.

Fig. \ref{fig6} depicts the variation of the witness expectation value as a function of the dissipation parameter $\gamma_g$ with varying $N$. We observe that for low dissipation $\gamma_g<<1$, the positive witness expectation value implies a separable state and hence no many-body entanglement exists in the steady state. With increasing dissipation $\gamma_g>>1$, the dynamics allows the generation of entangled state which is indicated by the onset of the negative witness value. Near the saturation value $\gamma_{\text{sat}}$, one observes the maximum negative expectation value of the witness operator that asymptotically converges to value $-\frac{1}{2}$ corresponding to the cluster state.
We donot observe appreciable change in the witness expectation value with different system size $N$ once the dissipation value has been reached to the saturation limit.

\begin{figure}[t]
  \centering
  \includegraphics[width=\linewidth]{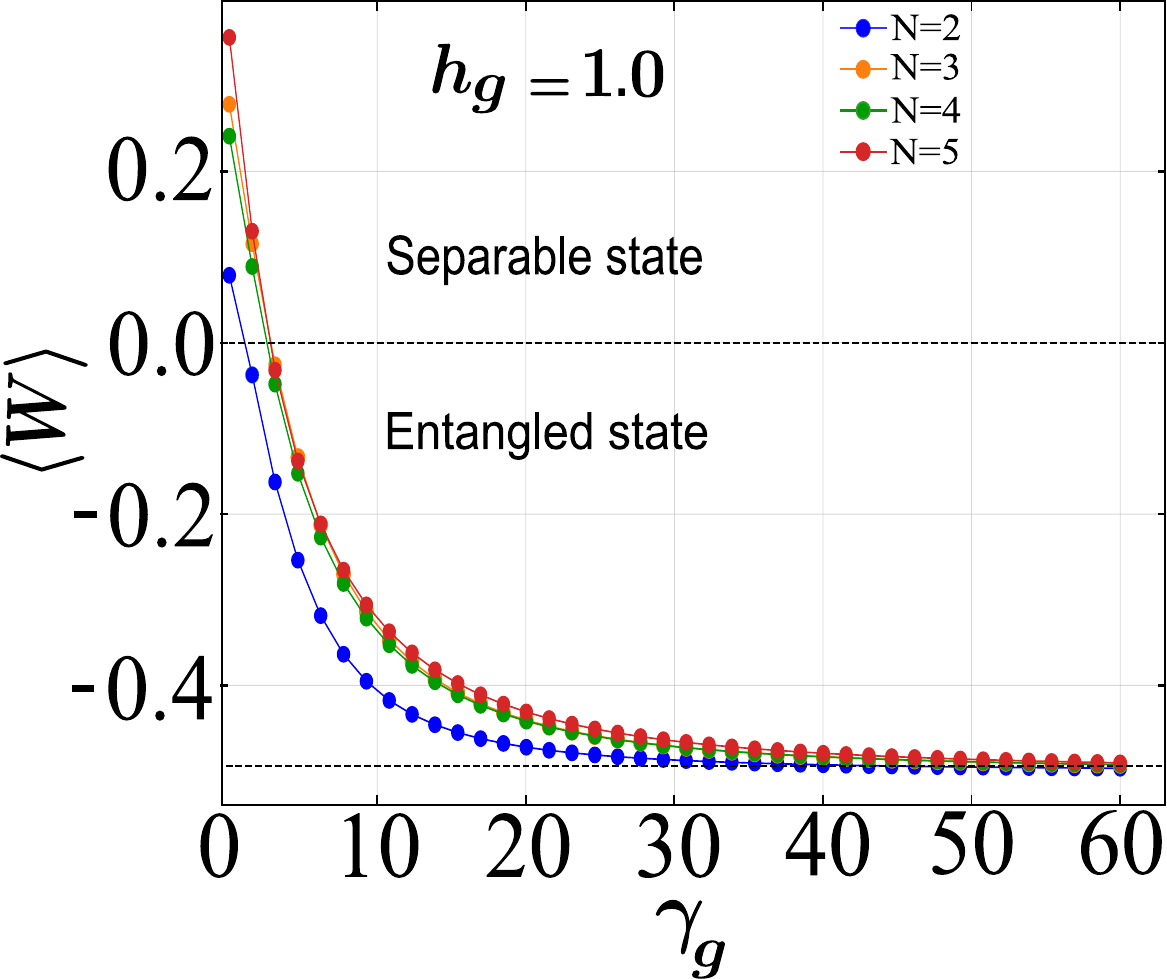}
  \caption{Variation of the expectation value of the entanglement witness \eqref{eq:witness_expectationvalue} as a function of the dissipative parameter $\gamma_g$ with different system size $N$ for the dissipative model \eqref{eq:masterequation} under Markovian dynamics. Dashed line at $\langle W \rangle = 0$ separates the region between separable and entangled states. Multipartite correlations for the cluster state is attained for the large dissipation regime $\gamma_g=\gamma_{\text{sat}}$ at which the witness expectation value saturates to the asymptotic value $-\frac{1}{2}$. We choose $h_g=1$. }
  \label{fig6}
\end{figure}

\subsection{Liouvillian gap}
The properties of the steady state density matrix can be obtained from the Liouvillian spectrum~\cite{minganti2018spectral}. For instance, we have calculated Liouvillian gap \eqref{eq:liou_gap} to reflect about the relaxation dynamics and nature of the transition among different steady states.

\begin{figure}[t]
  \centering
  \includegraphics[width=\linewidth]{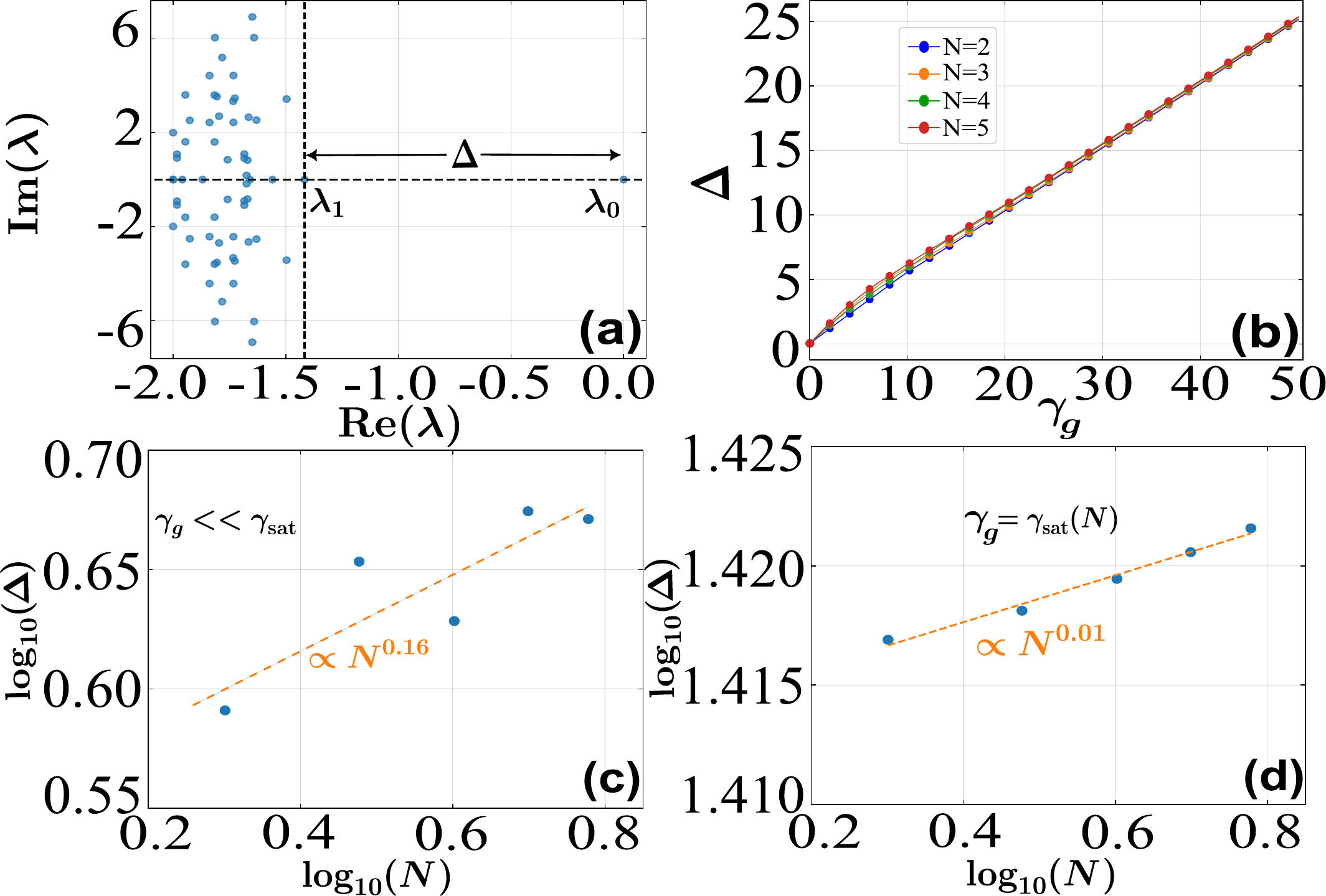}
\caption{Study of the spectral properties of the Liouvillian matrix~\eqref{eq:masterequation} and its scaling with the system size for the dissipative model under Markovian dynamics: (a) Eigenvalue spectrum~\eqref{eq:eigenval} of the Liouvillian matrix in the negative complex plane; (b) Variation of the gap~\eqref{eq:liou_gap} with the dissipative strength $\gamma_g$ for various system size $N$;  (c)-(d) Liouvillian gap is plotted as a function of the system size $N$ in a $\log{}$-$\log{}$ plot showing a power law scaling of the Liouvillian gap ($\propto N^{\beta}$) for weak and strong dissipation cases respectively. The system parameters are chosen as specified in each panel. For all cases, we choose $h_g=1$. }
  \label{fig7}
\end{figure}
Fig. \ref{fig7}(a) shows the eigenvalue spectrum for the Liouvillian \eqref{eq:masterequation} defined by Lindblad jump operators \eqref{eq:Lindbladope}. Eigenvalues exist in symmetrically placed complex conjugate pairs ($\text{Re}(\lambda),\pm\text{Im}(\lambda)$) and there is a unique steady state corresponding to the eigenvalue $\lambda_0$. One striking feature is that the eigenvalues are not cluttered near $\lambda_0$ implying a large gap. In Fig. \ref{fig7}(b), one observes that Liouvillian spectral gap $\Delta$ increases with the dissipation parameter ensuring faster approach to the target state and it remains finite even in the thermodynamic limit. We note the similar behavior in Liouvillian gap in Ref.~\cite{PhysRevB.109.064311} that has described it as a generic feature of chaotic open many-body systems. Moreover, one can see that the Liouvillian gap increases with $N$ slightly at a sufficiently small $\gamma_g<<\gamma_{\text{sat}}$, but is almost independent of $N$ for larger $\gamma_g\approx\gamma_{\text{sat}}$.

To characterize the dependence of the spectral gap $\Delta$ with the system size, we provide a scaling with $N$ by showing a straight line fit using a power law
\begin{align}
\Delta & \propto N^{\beta} \nonumber \\
    \log_{10}{\Delta} & = \beta \log_{10}{N} 
    \label{scaling_gap}
\end{align} 
In the low dissipation case $\gamma_g<<\gamma_{\text{sat}}$, the exponent has value $\beta=0.16$ remarking slow variation with $N$ as shown in Fig. \ref{fig7}(c). While, for the limiting case $\gamma_g\approx\gamma_{\text{sat}}$, it has further weaker dependence with $N$ where the exponent value is calculated as $\beta=0.01$ implying that the Liouvillian gap is insensitive to the system size (Fig. \ref{fig7}(d)) and there is no closing in the thermodynamic limit. This size-independent gap has interesting features~\cite{kastoryano2013rapid} such as a fast relaxing rate towards a fix steady state and absence of the long-range correlations in the other Liouvillian modes. This argument further supports the uniqueness of the stable cluster state that emerges in our dissipative model~\eqref{eq:masterequation} and  is a consequence of the global dissipation.

\section{Extension to the dissipative entanglement generation in high dimension}
\label{sec5}
In this section, we extend the idea of the dissipative model to generate cluster state in high dimension. As a prototype, we show the case for square lattice cluster state in two-dimension (2D) as this is a $\textit{universal}$ source for measurement based quantum computation that the linear cluster state in one dimension (1D) lacks. Various works on the generation of 2D cluster state have been proposed~\cite{PhysRevA.82.022331,larsen2019deterministic}. 

Mathematically, 2D square lattice cluster state is defined in a similar way~\eqref{eq:N-Cluster state} with the edges representing nearest neighbor interactions between four qubits ($1,2,3,4$) as 
\begin{align}
    \{(1,2), (1,3), (2,4), (3,4)\}
\end{align}
This can be visualized as two coupled linear cluster state placed side by side vertically as shown in Fig. \ref{fig1}(b) such that 
\begin{align}
    |C_{2\times 2}\rangle = CZ_{12} CZ_{13} CZ_{24} CZ_{34} |++++\rangle
    \label{eq:sqaure_cluster state}
\end{align}
with the full stabilizer group~\eqref{eq:stabilizer} 
\begin{align}
    \{X_1Z_2Z_3,X_2Z_1Z_4,X_3Z_1Z_4,X_4Z_2Z_3\}.
\end{align}
Using the same projection based Lindblad operator technique ~\eqref{eq:Lindbladope}, we prepare $4$-qubit square lattice cluster state as a unique steady state of the Liouvillian matrix.

We verify the multipartite correlations in the steady state density matrix of $4$-qubit square lattice cluster state using entanglement witness~\eqref{eq:witness}.  Fig. \ref{fig8}(a) shows the response of the witness expectation value with the dissipation. The variation follows the similar trend as observed in the case of a linear cluster state, i.e., for weak dissipation $\gamma_g<<1$, the witness remains positive, indicating a separable or weakly correlated steady state. As $\gamma_g$ is increased, $\langle W \rangle$ crosses zero and becomes negative, signaling the onset of genuine multipartite entanglement. We note that increasing dissipation stabilizes the steady state towards $4$-qubit square lattice cluster state only with a asymptotic negative value $\langle W\rangle = -\frac{1}{2}$. 

We have calculated the fidelity~\eqref{eq:fidelity} of the steady state density matrix with respect to the ideal 2D cluster state~\eqref{eq:sqaure_cluster state} as a function of the dissipation strength as shown in Fig.  \ref{fig8}(b). The fidelity increases monotonically with $\gamma_g$ starting from a moderate value at weak dissipation. The fidelity is maximum $F \sim 1$ as the limiting value of the dissipation is achieved at $\gamma_g=\gamma_{\text{sat}}$. This confirms that dissipative dynamics actively purify the steady state toward the target 2D cluster state only.
The negative entanglement witness and the near-unity fidelity at large $\gamma_g$ show the applicability of our protocol in generating high dimensional cluster state.

\begin{figure}[t]
  \centering
  \includegraphics[width=\linewidth]{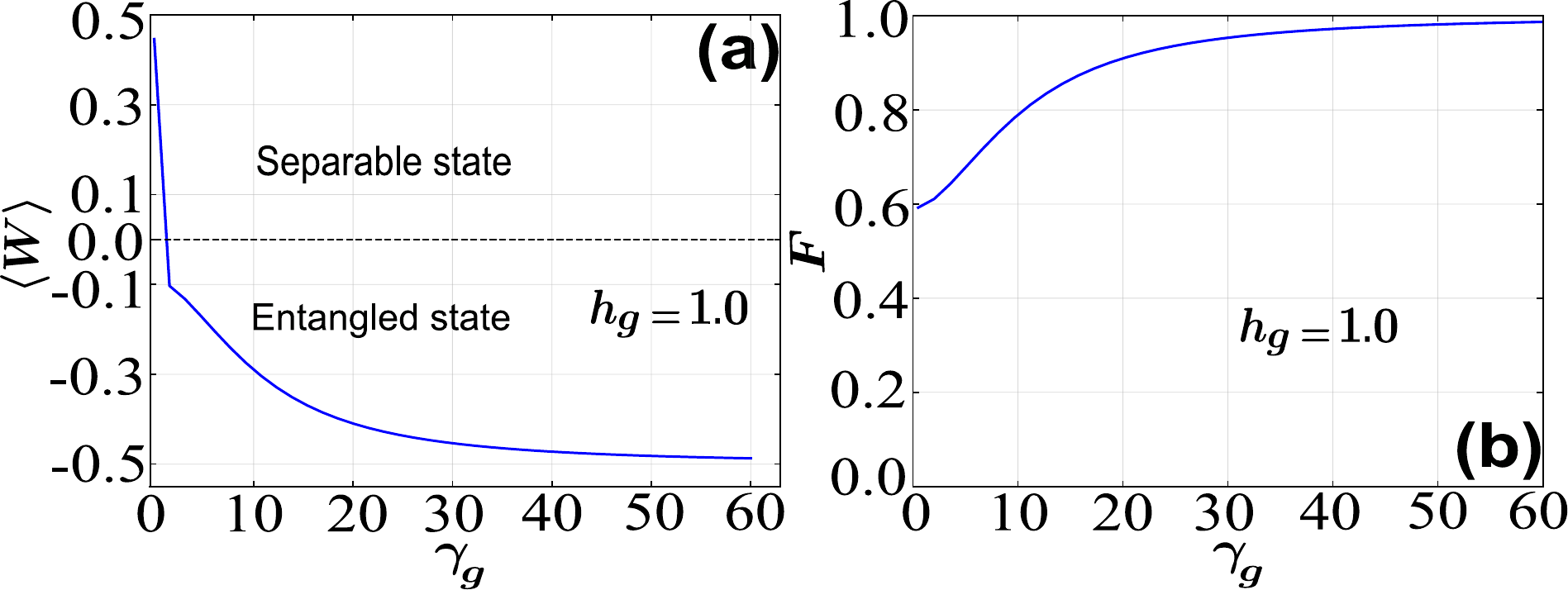}
  \caption{Variation of (a) entanglement witness~\eqref{eq:witness_expectationvalue} and (b) fidelity~\eqref{eq:fidelity} of the steady state with respect to the ideal 2D cluster state~\eqref{eq:sqaure_cluster state} as a function of the dissipation strength for the dissipative model \eqref{eq:masterequation} under Markovian dynamics. We choose $h_g=1$.}
  \label{fig8}
\end{figure}

\section{Experimental realization}
\label{sec6}
In this section, we discuss the viability of our dissipative protocol in terms of its experimental implementation. A potential choice for physical system to simulate our protocol is highly controllable trapped ion system of Yb${^+}$ or Sr${^+}$ ions~\cite{RevModPhys.93.025001}. The hyperfine ground states of the ions~\cite{jamesQuantumDynamicsCold1998} is encoded as an effective qubit such that each qubit state can be initialized in a particular spin configuration using optical pumping~\cite{RevModPhys.44.169}. Ion-ion interactions can be tailored by coupling the ions with a collective vibrational mode or phonon leading to a spin dependent force~\cite{PhysRevLett.94.153602}. There have been various experimental demonstrations to implement Ising interactions with trapped ions in an external magnetic field~\cite{PhysRevLett.134.050602,britton2012engineered}. This system is a versatile platform to study dissipative effects~\cite{PRXQuantum.3.010347,PhysRevLett.128.080502,PhysRevLett.115.200502} on the ion state.

The major challenge is to design a dissipative channel that implements Lindblad operator~\eqref{eq:Lindbladope} i.e., to obtain a dissipative map that moves the population from the orthogonal eigenspace into the desired eigenspace of the Liouvillian matrix $\mathcal{L}$.  As we have employed a state projection based construction, in principle any dissipative bath, that can alter the spin states of the ions, is a viable choice~\cite{PhysRevLett.115.200502}. One can design ancilla mediated autonomous pumping that flips the ancilla qubit conditional on the system being in the  orthogonal eigenspace of $\mathcal{L}$. The ancilla qubit is coupled to a lossy bath so that when it is flipped it decays back to the ground state while performing a corrective action on the system via interaction or conditional unitary $U$. One can apply a sequence of $U$ so that after ancilla decay, the system ends in the corrected subspace provided ancilla decay rate is larger than the system qubit decay rates. If the ancilla qubit flips on the undesired subspace, the ancilla decay is accompanied by an operation that corrects the system (correction conditioned on ancilla state). By adiabatic eliminating ancilla qubit, one obtains an effective Lindblad jump operator on the system of the form~\eqref{eq:Lindbladope}.  This ensures a Markovian dissipative environment for the system that we have considered. 

Based on the size of the linear or square cluster state, the dissipative reservoir can be engineered that couples with all the ions in an identical way leading to the same decay process. 


\section{Conclusion}
\label{sec7}
We have introduced a dissipative protocol based on Markovian evolution of multipartite system that have a desired cluster state as it's unique stable steady state.  This theoretical analysis considers the interplay between coherent interactions set by Ising couplings between qubits and the global dissipation acting on the qubits. Dissipation serves as an efficient resource for preparing and stabilizing highly entangled cluster states. We design Lindblad jump operators~\eqref{eq:Lindbladope} that removes the contribution of the orthogonal subspace by pumping the population into the desired cluster state. This approach splits the Hilbert space into two parts such dissipation process allows the convergence towards the target structure.  To perform theoretical analysis on the stable steady states for the dissipative model~\eqref{eq:masterequation}, we perform mean field calculation in the thermodynamical limit that factorizes the product of the spin operators. It predicts the existence of two types of steady states such that the cluster state emerges as the unique steady state of the engineered Liouvillian in the large dissipation limit.

Our numerical result for finite system size reveals that the fidelity approaches the asymptotic unit value, indicating a large overlap of the steady state with the target linear cluster state,   with the increasing dissipation strength at the critical spin-spin coupling. Similarly, the multipartite correlations of the steady state of the Liouvillian resembles with that of the ideal cluster state in the large dissipation limit. This is further supported by the finite Liouvillian gap that doesnot close in the large $N$ limit. We also note that at the threshold value $\gamma_{\text{sat}}$ that scales as $\sim N$, all these quantities are independent on the size of the cluster state offering insights into practical realization of the steady-state entanglement generation in the high dimensions.  

In our work, we construct Lindblad operator so as to pump the population into the target steady state that has practical implementation. It can be visualized as the mapping on to the particular spin or energy state.  However it requires $2^N-1$ number of such jump operators to define the full dynamics, though with larger $N$, one can always work with the subset of such operators to determine the dissipative dynamics. An alternative choice is to use the local Lindblad operators constructing from the  stabilizers~\eqref{eq:stabilizer} of the cluster state, i.e. $L_m=(I-\mathcal{S}_m)/2$ that offers better scalability $\sim N$ at the price of reduced fidelity. This produces dephasing in the stabilizer eigenbasis that leads to pumping into the undesired stabilizer subspace. Moreover these operators constitute the appearance of many steady states including the maximally mixed state as well.

The framework can be extended to larger cluster networks and can be implemented in trapped ion system with controllable dissipation. We propose an example for 2D square cluster state in our work and the current protocol is found to be robust in generating stabilized entanglement. Our projection based construction is efficient in driving the system dynamics towards the unique multipartite cluster state over the works~\cite{PhysRevE.73.016139,gonzalez2024tutorial} where such construction is performed at the level of deriving Master equation. Our demonstration of entanglement generation with asymptotic unit fidelity near the threshold dissipation advances the dissipative production of entangled resource states and computation.

\section*{Acknowledgment}
M.C. acknowledges the funding support from the FWO and the F.R.S.-FNRS as part of the Excellence of Science program (EOS project 40007526) at University of Liège, Belgium.

\appendix
\section{Construction of orthogonal subspace and Lindblad operators}
Let's construct the basis set of orthonormal subspace to the cluster state $|C_2\rangle$.
For $N=2$, the computational basis is given by
\begin{align}
    \{|00\rangle, |01\rangle, |10\rangle, |11\rangle \}
\end{align}
One convenient choice for orthonormal basis is 
\begin{align}
   \mathcal{H}_{\perp} =  \text{span}\{Z_1 |C_2\rangle, Z_2 |C_2\rangle, Z_1 Z_2 |C_2\rangle \}
\end{align}
Any $|\phi_m\rangle \in \mathcal{H}_{\perp}$ is orthogonal to the cluster state $|C_2\rangle$. The dimension of the orthogonal subspace is $3$. The choice is not unique and one can construct another basis set using Gram Schmidt orthonormalization process.

The construction of the Lindblad operators follows as,
\begin{align}
    L_1& = |C_2\rangle \langle \phi_1|\nonumber \\ 
    L_2& = |C_2\rangle \langle \phi_2|\nonumber \\
    L_3& = |C_2\rangle \langle \phi_3|
\end{align}
For any $N$, we can generalize the basis set for the orthogonal subspace as,
\begin{align}
    \mathcal{H}_{\perp} =  \text{span}\{ \prod_{m}  Z_m^s |C_N\rangle \}, s\in \{0,1 \}^N
\end{align}
with the dimension $2^{N}-1$.
\bibliographystyle{quantum}
\bibliography{ref}
\end{document}